\documentclass[12pt]{article}    

\usepackage[margin=0.75in]{geometry} 
\usepackage{amsmath,amssymb,amsthm}

\newtheorem{theorem}{Theorem}[section]
\newtheorem{proposition}[theorem]{Proposition}
\newtheorem{lemma}[theorem]{Lemma}

\newcommand{\be}{\begin{equation}}
\newcommand{\ee}{\end{equation}}

\newcommand{\bea}{\begin{eqnarray}}
\newcommand{\eea}{\end{eqnarray}}

\numberwithin{equation}{section}
\begin{document}

\title{Painlev\'{e} VI, Painlev\'{e} III and the Hankel Determinant Associated with a Degenerate Jacobi Unitary Ensemble}
\author{Chao Min\thanks{School of Mathematical Sciences, Huaqiao University, Quanzhou 362021, China; e-mail: chaomin@hqu.edu.cn}\: and Yang Chen\thanks{Department of Mathematics, Faculty of Science and Technology, University of Macau, Macau, China; e-mail: yangbrookchen@yahoo.co.uk}}


\date{\today}
\maketitle
\begin{abstract}
This paper studies the Hankel determinant generated by a perturbed Jacobi weight, which is closely related to the largest and smallest eigenvalue distribution of the degenerate Jacobi unitary ensemble. By using the ladder operator approach for the orthogonal polynomials, we find that the logarithmic derivative of the Hankel determinant satisfies a nonlinear second-order differential equation, which turns out to be the Jimbo-Miwa-Okamoto $\sigma$-form of the Painlev\'{e} VI equation by a translation transformation. We also show that, after a suitable double scaling, the differential equation is reduced to the Jimbo-Miwa-Okamoto $\sigma$-form of the Painlev\'{e} III. In the end, we obtain the asymptotic behavior of the Hankel determinant as $t\rightarrow1^{-}$ and $t\rightarrow0^{+}$ in two important cases, respectively.
\end{abstract}

$\mathbf{Keywords}$: Random matrix theory; Hankel determinant; Degenerate Jacobi unitary ensemble;

Ladder operators; Painlev\'{e} equations; Double scaling analysis.

$\mathbf{Mathematics\:\: Subject\:\: Classification\:\: 2010}$: 15B52, 42C05, 33E17.

\section{Introduction}
As a fundamental research object in random matrix theory, the Hankel determinant is defined by
$$
D_{n}:=\det\left(\int_{a}^{b}x^{i+j}w(x)dx\right)_{i,j=0}^{n-1},
$$
where $w(x)$ is a weight function supported on the interval $[a,b]$. This Hankel determinant is equal to the partition function of the unitary ensemble of $n\times n$ Hermitian matrices with eigenvalues $\{x_{1}, x_{2}, \ldots, x_{n}\}$ \cite{Mehta},
$$
D_{n}=\frac{1}{n!}\int_{[a,b]^{n}}\prod_{1\leq i<j\leq n}(x_{i}-x_{j})^{2}\prod_{j=1}^{n}w(x_{j})dx_{j}.
$$

Due to the importance in random matrix theory, Hankel determinants generated by the perturbed Gaussian, Laguerre and Jacobi weights have been studied extensively in the past decade, see \cite{Basor2015,Basor2010,Bogatskiy,ChenDai,ChenIts,Dai,Lyu2019,MinLyuChen,Xu2015,Xu2016,Zeng}. They can be viewed as the partition function of the corresponding perturbed unitary ensembles. Hankel determinants also play an important role in the gap probability problems, including the largest and smallest eigenvalue distribution in the unitary ensembles, see \cite{Basor2009,Basor2012,ChenZhang,Lyu2018,Min2018,Min2019a}. In addition, they have been applied to the wireless communication systems \cite{Chen2013,ChenMcKay2012} and many branches of applied mathematics and mathematical physics (see \cite{BCM,CJ,Jimbo1980,McCoy} for example).

In this paper, we consider the Hankel determinant generated by a perturbed Jacobi weight, namely,
\bea\label{dnt}
D_{n}(t):&=&\det\left(\int_{0}^{1}x^{i+j}w(x,t)dx\right)_{i,j=0}^{n-1}\nonumber\\
&=&\frac{1}{n!}\int_{[0,1]^{n}}\prod_{1\leq i<j\leq n}(x_{i}-x_{j})^{2}\prod_{j=1}^{n}w(x_{j},t)dx_{j},
\eea
where
$$
w(x,t):=x^{\alpha}(1-x)^{\beta}|x-t|^{\gamma}(A+B\theta(x-t)),\;\; x,\;t\in [0,1],\;\alpha,\;\beta,\;\gamma>0.
$$
Here $\theta(x)$ is the Heaviside step function, i.e., $\theta(x)$ is 1 for $x>0$ and 0 otherwise; $A$ and $B$ are constants, and $A\geq 0,\; A+B\geq 0$.

As stated in the previous work \cite{Min2019b}, this Hankel determinant is connected with the partition function of the Jacobi unitary ensemble with a single eigenvalue degeneracy, and the largest and smallest eigenvalue distribution of this degenerate Jacobi unitary ensemble. See also \cite{Wu} on the study of the Hankel determinant for the Gaussian case.

As a matter of fact, the Hankel determinants generated by perturbed Gaussian, Laguerre and Jacobi weights are usually related to the well-known nonlinear second-order ordinary differential equations-Painlev\'{e} equations. Chen and Zhang \cite{ChenZhang} studied the $\gamma=0$ case in (\ref{dnt}), and showed that the logarithmic derivative of the related Hankel determinant satisfies a particular $\sigma$-form of Painlev\'{e} VI. Dai and Zhang \cite{Dai} considered the situation of $t<0$ and $A=1, B=0$ in (\ref{dnt}), and established the relation of the Hankel determinant with another particular Painlev\'{e} VI. The main difference of our problem and \cite{Dai} is that our weight vanishes at a singular point $t$ in the interior of the support, and this will increase the difficulty in the analysis. Actually, they mentioned the case with $A=1, B=0$ in (\ref{dnt}), but did not prove it (see Remark 1.2 in \cite{Dai}). We will show that the results in \cite{Dai} are still valid in our problem.

The approach in this paper is the ladder operators associated with the orthogonal polynomials. We first write the Hankel determinant as the product of the square of the norm of the corresponding monic orthogonal polynomials (see (\ref{hankel})). Then we use the ladder operator approach to analyze the properties of the orthogonal polynomials to obtain a series of difference and differential equations. From this we establish the relation of the Hankel determinant with the Painlev\'{e} equations. This approach has been widely applied to the Hankel determinants for various perturbed weight functions, see \cite{Basor2009,Basor2015,Basor2010,Basor2012,ChenDai,ChenIts,ChenZhang,Dai,Lyu2019,Min2018,Min2019a,Min2019b,MinLyuChen}. We will state some elementary facts about the orthogonal polynomials at first.

Let $P_{n}(x,t)$ be the monic polynomials of degree $n$ orthogonal with respect to the weight $w(x,t)$,
\be\label{ops}
\int_{0}^{1}P_{m}(x,t)P_{n}(x,t)w(x,t)dx=h_{n}(t)\delta_{mn},\;\;m, n=0,1,2,\ldots,
\ee
where
$$
P_{n}(x,t)=x^{n}+\mathrm{p}(n,t)x^{n-1}+\cdots.
$$
We will see that $\mathrm{p}(n,t)$, the coefficient of $x^{n-1}$, plays an important role in the following discussions.

The three-term recurrence relation of the orthogonal polynomials shows that \cite{Chihara,Szego}
\be\label{rr}
xP_{n}(x,t)=P_{n+1}(x,t)+\alpha_{n}(t)P_{n}(x,t)+\beta_{n}(t)P_{n-1}(x,t),
\ee
supplemented by the initial conditions
$$
P_{0}(x,t)=1,\;\;\beta_{0}(t)P_{-1}(x,t)=0.
$$
It is easy to see that
\be\label{al}
\alpha_{n}(t)=\mathrm{p}(n,t)-\mathrm{p}(n+1,t)
\ee
and
\be\label{be}
\beta_{n}(t)=\frac{h_{n}(t)}{h_{n-1}(t)}.
\ee
A telescopic sum of (\ref{al}) gives
\be\label{sum}
\sum_{j=0}^{n-1}\alpha_{j}(t)=-\mathrm{p}(n,t).
\ee
Finally, it is well known that \cite{Ismail}
\be\label{hankel}
D_{n}(t)=\prod_{j=0}^{n-1}h_{j}(t).
\ee

This paper is organized as follows. Sec. 2 applies the ladder operators and the associated compatibility conditions to the perturbed Jacobi weight to obtain some important equations. Sec. 3 mainly derives the Painlev\'{e} VI equation sasisfied by the logarithmic derivative of the associated Hankel determinant. Sec. 4 shows that our problem is related to the Painlev\'{e} III equation under a suitable double scaling. In addition, the asymptotic behavior of the Hankel determinant is given in two special cases.

\section{Ladder Operators and its Compatibility Conditions}
In the following discussions, for convenience, we shall not display the $t$ dependence in $P_{n}(x)$, $w(x)$, $h_{n}$, $\alpha_{n}$ and $\beta_{n}$ unless it is needed. The following lemmas have been proved by Min and Chen \cite{Min2019b}, following Basor and Chen \cite{Basor2009} and Chen and Feigin \cite{ChenFeigin}.
\begin{lemma}\label{lo}
Let $w_{0}(x)$ be a smooth weight function defined on $[a,b]$, and $w_{0}(a)=w_{0}(b)=0$. The monic orthogonal polynomials with respect to $w(x):=w_{0}(x)|x-t|^{\gamma}(A+B\theta(x-t)),\; t\in[a,b],\; \gamma>0,\; A\geq 0,\; A+B\geq 0$ satisfy the lowering operator equation
\be\label{lowering}
\left(\frac{d}{dz}+B_{n}(z)\right)P_{n}(z)=\beta_{n}A_{n}(z)P_{n-1}(z),
\ee
and the raising operator equation
\be\label{raising}
\left(\frac{d}{dz}-B_{n}(z)-\mathrm{v}_{0}'(z)\right)P_{n-1}(z)=-A_{n-1}(z)P_{n}(z),
\ee
where
$$
A_{n}(z):=\frac{1}{h_{n}}\int_{a}^{b}\frac{\mathrm{v}_{0}'(z)-\mathrm{v}_{0}'(y)}{z-y}P_{n}^{2}(y)w(y)dy+a_{n}(z,t),
$$
$$
a_{n}(z,t):=\frac{\gamma}{h_{n}}\int_{a}^{b}\frac{P_{n}^{2}(y)}{(z-y)(y-t)}w(y)dy;
$$
$$
B_{n}(z):=\frac{1}{h_{n-1}}\int_{a}^{b}\frac{\mathrm{v}_{0}'(z)-\mathrm{v}_{0}'(y)}{z-y}P_{n}(y)P_{n-1}(y)w(y)dy+b_{n}(z,t),
$$
$$
b_{n}(z,t):=\frac{\gamma}{h_{n-1}}\int_{a}^{b}\frac{P_{n}(y)P_{n-1}(y)}{(z-y)(y-t)}w(y)dy
$$
and $\mathrm{v}_{0}(z)=-\ln w_{0}(z)$.
\end{lemma}

\begin{lemma}\label{s1s2}
The functions $A_{n}(z)$ and $B_{n}(z)$ satisfy the following equations:
\be
B_{n+1}(z)+B_{n}(z)=(z-\alpha_{n})A_{n}(z)-\mathrm{v}_{0}'(z), \tag{$S_{1}$}
\ee
\be
1+(z-\alpha_{n})(B_{n+1}(z)-B_{n}(z))=\beta_{n+1}A_{n+1}(z)-\beta_{n}A_{n-1}(z), \tag{$S_{2}$}
\ee
\be
B_{n}^{2}(z)+\mathrm{v}_{0}'(z)B_{n}(z)+\sum_{j=0}^{n-1}A_{j}(z)=\beta_{n}A_{n}(z)A_{n-1}(z). \tag{$S_{2}'$}
\ee
\end{lemma}

\begin{lemma}
The monic orthogonal polynomials $P_{n}(z)$ satisfy the second-order linear ordinary differential equation
$$
P_{n}''(z)-\left(\mathrm{v}_{0}'(z)+\frac{A_{n}'(z)}{A_{n}(z)}\right)P_{n}'(z)+\left(B_{n}'(z)-B_{n}(z)\frac{A_{n}'(z)}{A_{n}(z)}
+\sum_{j=0}^{n-1}A_{j}(z)\right)P_{n}(z)=0.
$$
\end{lemma}

\noindent $\mathbf{Remark\: 1.}$ Equations (\ref{lowering}) and (\ref{raising}) are called the ladder operators. The compatibility conditions ($S_{1}$), ($S_{2}$) and ($S_{2}'$) are valid for $z\in \mathbb{C}\cup\{\infty\}$.

For our problem, the weight and associated quantities are
$$
w(x)=w_{0}(x)|x-t|^{\gamma}(A+B\theta(x-t)),\;\; x,\;t\in[0,1],\;\gamma>0,
$$
$$
w_{0}(x)=x^{\alpha}(1-x)^{\beta},\;\;\mathrm{v}_{0}(x)=-\alpha\ln x-\beta\ln(1-x),\;\;\alpha,\;\beta>0,
$$
$$
\frac{\mathrm{v}_{0}'(z)-\mathrm{v}_{0}'(y)}{z-y}=\frac{\alpha}{zy}+\frac{\beta}{(1-z)(1-y)}.
$$
From Lemma \ref{lo} and noting that $w_{0}(0)=w_{0}(1)=0$, we have
\be\label{anz}
A_{n}(z)=\frac{x_{n}(t)}{z}-\frac{R_{n}(t)}{z-1}+a_{n}(z,t),
\ee
where
$$
x_{n}(t):=\frac{\alpha}{h_{n}}\int_{0}^{1}\frac{P_{n}^{2}(y)w(y)}{y}dy,
$$
$$
R_{n}(t):=\frac{\beta}{h_{n}}\int_{0}^{1}\frac{P_{n}^{2}(y)w(y)}{1-y}dy,
$$
$$
a_{n}(z,t)=\frac{\gamma}{h_{n}}\int_{0}^{1}\frac{P_{n}^{2}(y)w(y)}{(z-y)(y-t)}dy,
$$
and
\be\label{bnz}
B_{n}(z)=\frac{y_{n}(t)}{z}-\frac{r_{n}(t)}{z-1}+b_{n}(z,t),
\ee
where
$$
y_{n}(t):=\frac{\alpha}{h_{n-1}}\int_{0}^{1}\frac{P_{n}(y)P_{n-1}(y)w(y)}{y}dy,
$$
$$
r_{n}(t):=\frac{\beta}{h_{n-1}}\int_{0}^{1}\frac{P_{n}(y)P_{n-1}(y)w(y)}{1-y}dy,
$$
$$
b_{n}(z,t)=\frac{\gamma}{h_{n-1}}\int_{0}^{1}\frac{P_{n}(y)P_{n-1}(y)w(y)}{(z-y)(y-t)}dy.
$$

\begin{theorem}
As $z\rightarrow\infty$, $A_{n}(z)$ and $B_{n}(z)$ have the following series expansions:
\bea\label{anz1}
A_{n}(z)&=&\frac{\gamma+(t-1)R_{n}-t x_{n}}{z^2}+\frac{\gamma(t+\alpha_{n})+(t^2-1)R_n-t^2x_{n}}{z^3}\nonumber\\
&+&\frac{\gamma(t^2+\alpha_{n}^2+t\alpha_{n}+\beta_{n}+\beta_{n+1})+(t^3-1)R_n-t^3x_{n}}{z^4}+O\left(\frac{1}{z^5}\right),
\eea
\bea\label{bnz1}
B_{n}(z)&=&-\frac{n}{z}+\frac{(t-1) r_{n}-ty_{n}-nt}{z^2}+\frac{\gamma\beta_{n}-nt^2+(t^2-1) r_n-t^2 y_{n}}{z^3}\nonumber\\
&+&\frac{\gamma\beta_{n}(t+\alpha_{n}+\alpha_{n-1})+t^3(r_{n}-y_{n}-n)-r_n}{z^4}+O\left(\frac{1}{z^5}\right).
\eea
\end{theorem}
\begin{proof}
As $z\rightarrow\infty$,
$$
\frac{1}{z-1}=\frac{1}{z}\cdot\frac{1}{1-\frac{1}{z}}=\frac{1}{z}\left(1+\frac{1}{z}+\left(\frac{1}{z}\right)^2+\left(\frac{1}{z}\right)^3
+O\left(\frac{1}{z^4}\right)\right)
=\frac{1}{z}+\frac{1}{z^2}+\frac{1}{z^3}+\frac{1}{z^4}+O\left(\frac{1}{z^5}\right)
$$
and
$$
\frac{1}{z-y}=\frac{1}{z}\cdot\frac{1}{1-\frac{y}{z}}=\frac{1}{z}\left(1+\frac{y}{z}+\left(\frac{y}{z}\right)^2+\left(\frac{y}{z}\right)^3
+O\left(\frac{1}{z^4}\right)\right)
=\frac{1}{z}+\frac{y}{z^2}+\frac{y^2}{z^3}+\frac{y^3}{z^4}+O\left(\frac{1}{z^5}\right).
$$
Then from the definition of $a_n(z,t)$ and $b_n(z,t)$ we have
\bea
a_n(z,t)&=&\frac{\gamma}{zh_{n}}\int_{0}^{1}\frac{P_{n}^{2}(y)w(y)}{y-t}dy+\frac{\gamma}{z^2h_{n}}\int_{0}^{1}\frac{yP_{n}^{2}(y)w(y)}{y-t}dy
+\frac{\gamma}{z^3h_{n}}\int_{0}^{1}\frac{y^2P_{n}^{2}(y)w(y)}{y-t}dy\nonumber\\
&+&\frac{\gamma}{z^4h_{n}}\int_{0}^{1}\frac{y^3P_{n}^{2}(y)w(y)}{y-t}dy+O\left(\frac{1}{z^5}\right),\nonumber
\eea
\bea
b_n(z,t)&=&\frac{\gamma}{zh_{n-1}}\int_{0}^{\infty}\frac{P_{n}(y)P_{n-1}(y)w(y)}{y-t}dy+\frac{\gamma}{z^2h_{n-1}}\int_{0}^{\infty}
\frac{yP_{n}(y)P_{n-1}(y)w(y)}{y-t}dy\nonumber\\
&+&\frac{\gamma}{z^3h_{n-1}}\int_{0}^{\infty}\frac{y^2P_{n}(y)P_{n-1}(y)w(y)}{y-t}dy
+\frac{\gamma}{z^4h_{n-1}}\int_{0}^{\infty}\frac{y^3P_{n}(y)P_{n-1}(y)w(y)}{y-t}dy+O\left(\frac{1}{z^5}\right).\nonumber
\eea
Through integration by parts, we find
\bea
\alpha\int_{0}^{1}\frac{P_{n}^{2}(y)w(y)}{y}dy&=&\int_{0}^{1}P_{n}^{2}(y)(1-y)^{\beta}|y-t|^{\gamma}(A+B\theta(y-t))dy^{\alpha}\nonumber\\
&=&\beta\int_{0}^{1}\frac{P_{n}^{2}(y)w(y)}{1-y}dy-\gamma\int_{0}^{1}\frac{P_{n}^{2}(y)w(y)}{y-t}dy,
\eea
where we have used the formula \cite{ChenFeigin}
$$
\partial_{y}|y-t|^{\gamma}=\delta(y-t)((y-t)^{\gamma}-(t-y)^{\gamma})+\gamma\frac{|y-t|^{\gamma}}{y-t}
$$
and
$$
\partial_{y}\theta(y-t)=\delta(y-t).
$$
It follows that
\be\label{imp1}
\frac{\gamma}{h_{n}}\int_{0}^{1}\frac{P_{n}^{2}(y)w(y)}{y-t}dy=R_{n}-x_{n}.
\ee
Similarly, we have
\be\label{imp2}
\frac{\gamma}{h_{n-1}}\int_{0}^{1}\frac{P_{n}(y)P_{n-1}(y)w(y)}{y-t}dy=r_{n}-y_{n}-n.
\ee
With the aid of the recurrence relation (\ref{rr}) and the orthogonality (\ref{ops}), we find
\bea
a_{n}(z,t)&=&\frac{R_{n}-x_{n}}{z}+\frac{\gamma+t(R_{n}- x_{n})}{z^2}+\frac{\gamma(t+\alpha_{n})+t^2(R_n-x_{n})}{z^3}\nonumber\\
&+&\frac{\gamma(t^2+\alpha_{n}^2+t\alpha_{n}+\beta_{n}+\beta_{n+1})+t^3(R_n-x_{n})}{z^4}+O\left(\frac{1}{z^5}\right),\nonumber
\eea
\bea
b_{n}(z,t)&=&\frac{r_{n}-y_{n}-n}{z}+\frac{t (r_{n}-y_{n}-n)}{z^2}+\frac{\gamma\beta_{n}+t^2(r_{n}-y_{n}-n)}{z^3}\nonumber\\
&+&\frac{\gamma\beta_{n}(t+\alpha_{n}+\alpha_{n-1})+t^3(r_{n}-y_{n}-n)}{z^4}+O\left(\frac{1}{z^5}\right).\nonumber
\eea
In view of (\ref{anz}) and (\ref{bnz}), we obtain the desired results.
\end{proof}
Substituting (\ref{anz1}) and (\ref{bnz1}) into ($S_{1}$), we get the following three equations by comparing the coefficients of $\frac{1}{z}$, $\frac{1}{z^2}$ and $\frac{1}{z^3}$ on both sides:
\be\label{s11}
(t-1) R_n-tx_{n}+\alpha+\beta+\gamma+2n+1=0,
\ee
\be\label{s12}
(t-1)(r_{n+1}+r_{n})-t(y_{n+1}+y_{n})=\beta+(2n+1+\gamma)t+(t^2-\alpha_{n}t+\alpha_{n}-1)R_{n}+(t\alpha_{n}-t^2)x_{n},
\ee
\be\label{s13}
(t^2-1)(r_{n+1}+r_{n})-t^2(y_{n+1}+y_{n})=\beta+(2n+1+\gamma)t^2+(t^3-\alpha_{n}t^2+\alpha_{n}-1)R_{n}+(t^2\alpha_{n}-t^3)x_{n}.
\ee
Eliminating $y_{n+1}+y_{n}$ from (\ref{s12}) and (\ref{s13}), we have
\be\label{s14}
r_{n+1}+r_{n}=(1-\alpha_n)R_n-\beta.
\ee
Similarly, eliminating $r_{n+1}+r_{n}$ from (\ref{s12}) and (\ref{s13}) gives
$$
y_{n+1}+y_{n}+(\alpha_{n}-t)x_{n}+(t-1)R_n+\beta+\gamma+2n+1=0.
$$
Using (\ref{s11}), it follows that
\be\label{s15}
y_{n+1}+y_{n}=\alpha-\alpha_{n}x_{n}.
\ee

Similarly, substituting (\ref{anz1}) and (\ref{bnz1}) into ($S_{2}$), we have the following equations by comparing the coefficients of $\frac{1}{z}$, $\frac{1}{z^2}$ and $\frac{1}{z^3}$:
\be\label{s21}
\alpha_{n}=t-(t-1)(r_{n+1}-r_{n})+t(y_{n+1}-y_{n}),
\ee
\bea\label{s22}
&&(t-1)(\beta_{n+1}R_{n+1}-\beta_{n}R_{n-1})-t(\beta_{n+1}x_{n+1}-\beta_{n}x_{n-1})\nonumber\\
&=&(t^2-\alpha_{n}t+\alpha_{n}-1)(r_{n+1}-r_{n})-(t^2-t\alpha_{n})(y_{n+1}-y_{n})+t\alpha_{n}-t^2,
\eea
\bea\label{s23}
&&(t^2-1)(\beta_{n+1}R_{n+1}-\beta_{n}R_{n-1})-t^2(\beta_{n+1}x_{n+1}-\beta_{n}x_{n-1})\nonumber\\
&=&(t^3-\alpha_{n}t^2+\alpha_{n}-1)(r_{n+1}-r_{n})-(t^3-t^2\alpha_{n})(y_{n+1}-y_{n})+t^2\alpha_{n}-t^3.
\eea
In view of (\ref{sum}), a telescopic sum of (\ref{s21}) gives
\be\label{s24}
\mathrm{p}(n,t)=(t-1)r_{n}-ty_{n}-nt.
\ee

From the above results, we have three important equations in the following proposition.
\begin{proposition}
\be\label{s26}
\beta_{n}R_{n}R_{n-1}=r_{n}^2+\beta r_{n},
\ee
\be\label{s27}
\beta_{n}x_{n}x_{n-1}=y_{n}^2-\alpha y_{n},
\ee
\be\label{s210}
\beta_{n}(R_{n}x_{n-1}+x_{n}R_{n-1})=2r_{n}y_{n}+(2n+\beta+\gamma)r_{n}-(2n+\alpha+\gamma)y_{n}-n(n+\gamma).
\ee
\end{proposition}
\begin{proof}
We start from eliminating $\beta_{n+1}x_{n+1}-\beta_{n}x_{n-1}$ from (\ref{s22}) and (\ref{s23}) to obtain
\be\label{s25}
\beta_{n+1}R_{n+1}-\beta_{n}R_{n-1}=(1-\alpha_{n})(r_{n+1}-r_{n}).
\ee
Multiplying both sides of (\ref{s25}) by $R_{n}$ produces
$$
\beta_{n+1}R_{n+1}R_{n}-\beta_{n}R_{n}R_{n-1}=(1-\alpha_{n})R_{n}(r_{n+1}-r_{n}).
$$
Using (\ref{s14}), the above becomes
\be\label{eqs2}
\beta_{n+1}R_{n+1}R_{n}-\beta_{n}R_{n}R_{n-1}=r_{n+1}^2-r_{n}^2+\beta(r_{n+1}-r_{n}).
\ee
Noting that $\beta_{0}R_{0}(t)R_{-1}(t)=0$ and $r_{0}(t)=0$, a telescopic sum of (\ref{eqs2}) gives (\ref{s26}).

To proceed, we substitute (\ref{s25}) into (\ref{s22}) and obtain
$$
\beta_{n+1}x_{n+1}-\beta_{n}x_{n-1}=-(t-1)(r_{n+1}-r_{n})+(t-\alpha_{n})(y_{n+1}-y_{n})+t-\alpha_{n}.
$$
Using (\ref{s21}), the above becomes
$$
\beta_{n+1}x_{n+1}-\beta_{n}x_{n-1}=-\alpha_{n}(y_{n+1}-y_{n}).
$$
Multiplying both sides by $x_{n}$ gives
$$
\beta_{n+1}x_{n+1}x_{n}-\beta_{n}x_{n}x_{n-1}=-\alpha_{n}x_{n}(y_{n+1}-y_{n}).
$$
It follows from (\ref{s15}) that
$$
\beta_{n+1}x_{n+1}x_{n}-\beta_{n}x_{n}x_{n-1}=y_{n+1}^2-y_{n}^2-\alpha(y_{n+1}-y_{n}).
$$
A telescopic sum produces (\ref{s27}).

To derive (\ref{s210}), we subtract (\ref{s22}) from (\ref{s23}) to obtain
$$
\beta_{n+1}(R_{n+1}-x_{n+1})-\beta_{n}(R_{n-1}-x_{n-1})=(t-\alpha_{n})(r_{n+1}-r_n-y_{n+1}+y_{n}-1).
$$
Multiplying both sides by $R_{n}-x_{n}$ shows that
\be\label{s28}
\beta_{n+1}(R_{n+1}-x_{n+1})(R_{n}-x_{n})-\beta_{n}(R_{n}-x_{n})(R_{n-1}-x_{n-1})=(t-\alpha_{n})(R_{n}-x_{n})(r_{n+1}-r_n-y_{n+1}+y_{n}-1).
\ee
The difference of (\ref{s13}) and (\ref{s12}) gives
$$
(t-\alpha_{n})(R_{n}-x_{n})=r_{n+1}+r_n-y_{n+1}-y_{n}-2n-1-\gamma.
$$
Inserting it into (\ref{s28}), we find
\bea
&&\beta_{n+1}(R_{n+1}-x_{n+1})(R_{n}-x_{n})-\beta_{n}(R_{n}-x_{n})(R_{n-1}-x_{n-1})\nonumber\\
&=&(r_{n+1}-y_{n+1}-n-1-\gamma)(r_{n+1}-y_{n+1}-n-1)-(r_{n}-y_{n}-n-\gamma)(r_{n}-y_{n}-n).\nonumber
\eea
A telescopic sum produces
\be\label{s29}
\beta_{n}(R_{n}-x_{n})(R_{n-1}-x_{n-1})=(r_{n}-y_{n}-n-\gamma)(r_{n}-y_{n}-n).
\ee
Substituting (\ref{s26}) and (\ref{s27}) into (\ref{s29}), we obtain (\ref{s210}).
\end{proof}
\noindent $\mathbf{Remark\: 2.}$ The three important equations (\ref{s26}), (\ref{s27}) and (\ref{s210}) were derived from ($S_{2}'$) directly in \cite{Dai}. But in our case, we have to obtain them from the combination of ($S_{1}$) and ($S_{2}$).

Finally, substituting (\ref{anz1}) and (\ref{bnz1}) into ($S_{2}'$) and comparing the coefficients of $\frac{1}{z^2}$ on both sides, we get (\ref{s11}). Comparing the coefficients of $\frac{1}{z^3}$, we find the following equation,
\be\label{s31}
n\beta+n(2n+\alpha+\beta+\gamma)t-(2n+\alpha+\beta)(t-1)r_n+(2n+\alpha+\beta)t y_{n}+\gamma\sum_{j=0}^{n-1}\alpha_{j}+(t^2-1)\sum_{j=0}^{n-1}R_{j}-t^2\sum_{j=0}^{n-1}x_{j}=0.
\ee
Using (\ref{sum}) and (\ref{s24}), we have
$$
\sum_{j=0}^{n-1}\alpha_{j}=-\mathrm{p}(n,t)=-(t-1)r_{n}+ty_{n}+nt.
$$
Plugging it into (\ref{s31}) and replacing $x_{j}$ by $R_j$ from (\ref{s11}), we obtain
\be\label{s32}
n\beta+n(n+\gamma) t-(2n+\alpha+\beta+\gamma)(t-1)r_n+(2n+\alpha+\beta+\gamma)ty_{n}+(t-1)\sum_{j=0}^{n-1}R_{j}=0.
\ee

The equations obtained from ($S_{1}$), ($S_{2}$) and ($S_{2}'$) will play a significant role in the derivation of the Painlev\'{e} VI equation satisfied by the logarithmic derivative of the Hankel determinant $D_n(t)$ in the next section.

\section{Painlev\'{e} VI and its $\sigma$-Form}
In this section, we devote our efforts to deriving the second-order ordinary differential equation satisfied by the logarithmic derivative of the Hankel determinant. After suitable transformation, the relation of our Hankel determinant with the Painlev\'{e} equation will be established.
We begin with taking a derivative with respect to $t$ in the following equation
$$
\int_{0}^{1}P_{n}^2(x,t)x^{\alpha}(1-x)^{\beta}|x-t|^{\gamma}(A+B\theta(x-t))dx=h_{n}(t),\;\;n=0,1,2,\ldots,
$$
which produces
$$
h_{n}'(t)=-\gamma\int_{0}^{1}\frac{P_{n}^2(x)w(x)}{x-t}dx.
$$
It follows from (\ref{imp1}) that
\be\label{d1}
t\frac{d}{dt}\ln h_{n}(t)=tx_{n}-tR_n=\alpha+\beta+\gamma+2n+1-R_n,
\ee
where we have used (\ref{s11}) in the second equality.\\
In view of (\ref{be}), we have
$$
t\frac{d}{dt}\ln \beta_{n}(t)=2-R_{n}+R_{n-1}.
$$
Then,
\be\label{d11}
t\beta_{n}'(t)=\beta_{n}(2-R_{n}+R_{n-1}).
\ee
We now define a quantity related to the logarithmic derivative of the Hankel determinant,
$$
H_{n}(t):=t(t-1)\frac{d}{dt}\ln D_{n}(t).
$$
From (\ref{hankel}) and (\ref{d1}) we have
\be\label{hn}
H_{n}(t)=n(n+\alpha+\beta+\gamma)(t-1)-(t-1)\sum_{j=0}^{n-1}R_j(t).
\ee

On the other hand, taking a derivative with respect to $t$ in the equation
$$
\int_{0}^{1}P_{n}(x,t)P_{n-1}(x,t)x^{\alpha}(1-x)^{\beta}|x-t|^{\gamma}(A+B\theta(x-t))dx=0,\;\;n=0,1,2,\ldots,
$$
we obtain
\be\label{d2}
\frac{d}{dt}\mathrm{p}(n,t)=r_{n}-y_{n}-n.
\ee
It follows that
\be\label{alphap}
t\alpha_{n}'(t)=\alpha_{n}+r_n-r_{n+1},
\ee
where use has been made of (\ref{al}) and (\ref{s21}).\\
In addition, substituting (\ref{s24}) into (\ref{d2}) gives the following important relation,
\be\label{re}
(t-1)r_n'(t)=ty_{n}'(t).
\ee
\begin{proposition}\label{albe}
The recurrence coefficient $\alpha_{n}$ has the following expression in terms of $r_{n}, y_{n}$ and $R_{n}$,
\be\label{alpha}
(\alpha+\beta+\gamma+2n+2)\alpha_{n}=2(t-1)r_{n}-2ty_{n}-(t-1)R_{n}+t\alpha+(t-1)\beta+t,
\ee
while $\beta_{n}$ has the expression in terms of $r_{n}$ and $y_{n}$,
\bea\label{beta}
&&(\alpha+\beta+\gamma+2n+1)(\alpha+\beta+\gamma+2n-1)\beta_{n}\nonumber\\
&=&\left[ty_{n}-(t-1)r_{n}\right]^{2}-(t-1)(2nt+\gamma t+\beta)r_{n}+t\left[(t-1)(2n+\gamma)-\alpha\right]y_{n}+n(n+\gamma)(t^2-t).\nonumber\\
\eea
\end{proposition}
\begin{proof}
From (\ref{s14}) and (\ref{s15}) we have
\be\label{eq1}
r_{n+1}=(1-\alpha_n)R_n-\beta-r_{n}
\ee
and
\be\label{eq2}
y_{n+1}=\alpha-\alpha_{n}x_{n}-y_{n}.
\ee
Substituting (\ref{eq1}) and (\ref{eq2}) into (\ref{s21}), and using (\ref{s11}) to eliminate $x_{n}$, we obtain (\ref{alpha}).

Multiplying $(\ref{s27})$ by $t^2$ on both sides, and using (\ref{s11}) to eliminate $x_{n}$ and $x_{n-1}$, together with the aid of (\ref{s26}), we obtain
\bea\label{eq3}
&&(t-1)\left[(\alpha+\beta+\gamma+2n-1)\beta_n R_n+(\alpha+\beta+\gamma+2n+1)\beta_n R_{n-1}\right]\nonumber\\
&=&t^2(y_{n}^2-\alpha y_{n})-(t-1)^2(r_n^2+\beta r_n)-(\alpha+\beta+\gamma+2n+1)(\alpha+\beta+\gamma+2n-1)\beta_n.
\eea
Similarly, from (\ref{s210}) we have
\bea\label{eq4}
&&(\alpha+\beta+\gamma+2n-1)\beta_n R_n+(\alpha+\beta+\gamma+2n+1)\beta_n R_{n-1}\nonumber\\
&=&(2n+\beta+\gamma)t r_n-(2n+\alpha+\gamma)ty_{n}+2tr_{n}y_{n}-2(t-1)(r_n^2+\beta r_n)-n(n+\gamma)t.
\eea
Substituting (\ref{eq4}) into (\ref{eq3}) we obtain (\ref{beta}).
\end{proof}

\begin{proposition}\label{rrs}
The auxiliary quantities $r_n(t)$ and $y_{n}(t)$ can be expressed in terms of $H_n(t)$ and $H_n'(t)$ as follows,
$$
r_n(t)=\frac{n^2+n\alpha+n\gamma+H_n(t)-tH_n'(t)}{2n+\alpha+\beta+\gamma},
$$
$$
y_{n}(t)=-\frac{n^2+n\beta+n\gamma-H_n(t)+(t-1)H_n'(t)}{2n+\alpha+\beta+\gamma}.
$$
\end{proposition}
\begin{proof}
From (\ref{s32}) and (\ref{hn}) we have
\be\label{eq5}
(2n+\alpha+\beta+\gamma)\left[(t-1)r_n-ty_{n}\right]=n\left(2n+\alpha+\beta+2\gamma\right)t-n(n+\alpha+\gamma)-H_n(t).
\ee
Taking a derivative with respect to $t$ and using (\ref{re}), we obtain
\be\label{eq6}
(2n+\alpha +\beta +\gamma)(r_n-y_{n})=n (2n+\alpha +\beta +2\gamma)-H_n'(t).
\ee
The combination of (\ref{eq5}) and (\ref{eq6}) gives the desired result.
\end{proof}
From the above results, we are now able to prove the following theorem.
\begin{theorem}\label{thm}
The quantity $H_n(t)$ satisfies the following nonlinear second-order ordinary differential equation,
\bea\label{sod}
&&t^2(t-1)^2 (H_n'')^2+4t (t-1) (H_n')^3-\Big\{4 (2 t-1) H_n+\left[4 n^2+4 n (\alpha +\beta +\gamma )+(\alpha +\beta )^2\right]t^2\nonumber\\
&-&2\left[2 n^2+2 n (\alpha +\beta +\gamma)+\alpha(\alpha+\gamma)+\beta(\alpha-\gamma)\right]t+(\alpha+\gamma)^2\Big\}(H_n')^2\nonumber\\
&+&2\Big\{2H_n^2+\left[\left(4 n^2+4 n (\alpha +\beta +\gamma )+(\alpha +\beta )^2\right)t-2 n^2-2 n (\alpha +\beta +\gamma )-\alpha(\alpha+\gamma)-\beta(\alpha-\gamma)\right]H_n\nonumber\\
&-&n\gamma(n+\alpha +\beta +\gamma)  [(\alpha-\beta) t-\alpha-\gamma]\Big\}H_n'- \left[4 n^2+4 n (\alpha +\beta +\gamma )+(\alpha +\beta )^2\right]H_n^2\nonumber\\
&+&2 n\gamma(\alpha -\beta )(n+\alpha +\beta +\gamma)  H_n- n^2\gamma^2(n+\alpha +\beta +\gamma)^2=0.
\eea
\end{theorem}
\begin{proof}
Define
$$
X:=\beta_{n}R_{n},
$$
$$
Y:=\beta_{n}R_{n-1}.
$$
It follows from (\ref{s26}) that
\be\label{prod}
X\cdot Y=\beta_{n}(r_{n}^2+\beta r_{n}).
\ee
From (\ref{d11}) and (\ref{eq4}) we get two linear equations satisfied by $X$ and $Y$:
$$
X-Y=2\beta_n-t \beta_n',
$$
$$
(2n+\alpha+\beta+\gamma-1)X+(2n+\alpha+\beta+\gamma+1)Y=[2 \beta +(2n-\beta +\gamma)t]r_n-2 (t-1) r_n^2-(2n+\alpha +\gamma)ty_{n}+2tr_ny_{n}-n(n+\gamma)t.
$$
Solving for $X$ and $Y$ and substituting the expressions into (\ref{prod}), we get an equation for $r_n, y_{n}, \beta_n$ and $\beta_n'$. Using (\ref{beta}) and Proposition \ref{rrs}, we obtain the second-order differential equation satisfied by $H_n(t)$.
\end{proof}
From Theorem \ref{thm}, we readily have the following result.
\begin{theorem}\label{thm1}
Let
$$
\sigma_{n}(t):=H_n(t)+c_1t+c_2,
$$
where
$$
c_1=-n(n+\alpha+\beta+\gamma)-\frac{(\alpha+\beta)^2}{4},
$$
$$
c_2=\frac{1}{4}\left[2n(n+\alpha+\beta+\gamma)+(\alpha+\beta)\beta-(\alpha-\beta)\gamma\right].
$$
Then $\sigma_{n}(t)$ satisfies the following Jimbo-Miwa-Okamoto $\sigma$-form of Painlev\'{e} VI \cite{Jimbo1981} (page 446 (C.61)),
$$
\sigma_{n}'\left[t(t-1)\sigma_{n}''\right]^2+\left[2\sigma_{n}'(t\sigma_{n}'-\sigma_{n})-(\sigma_{n}')^2-\nu_1\nu_2\nu_3\nu_4\right]^2
=(\sigma_{n}'+\nu_1^2)(\sigma_{n}'+\nu_2^2)(\sigma_{n}'+\nu_3^2)(\sigma_{n}'+\nu_4^2),
$$
with the parameters
$$
\nu_1=\frac{\alpha+\beta}{2},\;\;\nu_2=\frac{\beta-\alpha}{2},\;\;\nu_3=\frac{2n+\alpha+\beta}{2},\;\;\nu_4=\frac{2n+\alpha+\beta+2\gamma}{2}.
$$
\end{theorem}
\begin{proof}
Substituting $H_n(t)=\sigma_{n}(t)-c_1t-c_2$ into equation (\ref{sod}), we obtain the desired result.
\end{proof}
In the end, we show that the auxiliary quantity $R_{n}(t)$ satisfies a particular painlev\'{e} VI up to a linear transformation.
\begin{theorem}\label{thm2}
Let
\be\label{tr}
S_n(t):=\frac{(t-1)R_n(t)}{2n+\alpha+\beta+\gamma+1}+1.
\ee
Then $S_n(t)$ satisfies the Painlev\'{e} VI equation \cite{Gromak}
\bea
S_n''&=&\frac{1}{2}\left(\frac{1}{S_n}+\frac{1}{S_n-1}+\frac{1}{S_n-t}\right)(S_n')^2-\left(\frac{1}{t}+\frac{1}{t-1}+\frac{1}{S_n-t}\right)S_n'\nonumber\\
&+&\frac{S_n(S_n-1)(S_n-t)}{t^2(t-1)^2}\left(\mu_1+\frac{\mu_2 t}{S_n^2}+\frac{\mu_3(t-1)}{(S_n-1)^2}+\frac{\mu_4 t(t-1)}{(S_n-t)^2}\right),\nonumber
\eea
where
$$
\mu_1=\frac{(2n+\alpha+\beta+\gamma+1)^2}{2},\;\;\mu_2=-\frac{\alpha^2}{2},\;\;\mu_3=\frac{\beta^2}{2},\;\;\mu_4=\frac{1-\gamma^2}{2}.
$$
\end{theorem}
\begin{proof}
Eliminating $r_{n+1}$ from (\ref{s14}) and (\ref{alphap}) gives
$$
t\alpha_{n}'(t)=\beta+\alpha_{n}+2r_n-(1-\alpha_{n})R_n.
$$
Substituting (\ref{alpha}) into the above and taking account of (\ref{re}), we obtain the expression of $y_{n}$ in terms of $R_n, R_n'$ and $r_n$,
\bea\label{rns}
y_{n}&=&\frac{1}{2 t R_n}\Big\{t(t-1) R_n'-(t-1) R_n^2+\left[2 (t-1) r_n+(\alpha +\beta +1)t-\alpha -2 \beta -\gamma -2 n-1\right]R_n\nonumber\\
&+& 2( \alpha + \beta + \gamma +2 n+1)r_n+\beta(\alpha +\beta+  \gamma +2  n+1)\Big\}.
\eea
Replacing $\beta_n R_{n-1}$ by $\frac{r_n^2+\beta r_n}{R_{n}}$ in (\ref{eq4}) and making use of (\ref{beta}) together with the aid of (\ref{rns}), we finally obtain a linear equation for $r_n$. Hence we can express $r_n$ in terms of $R_n$ and $R_n'$,
\be\label{rnex}
r_n=f(R_n, R_n'),
\ee
where $f(R_n, R_n')$ is a function that is explicitly known. We do not write it down because it is somewhat long.

On the other hand, using (\ref{s26}), equation (\ref{d11}) can be written as
\be\label{betap}
t\beta_{n}'(t)=\beta_{n}(2-R_{n})+\frac{r_n^2+\beta r_n}{R_{n}}.
\ee
In view of (\ref{beta}), (\ref{rns}) and (\ref{rnex}), we see that $\beta_n$ can also be exclusively expressed in terms of $R_n$ and $R_n'$,
\be\label{beex}
\beta_n=g(R_n, R_n'),
\ee
where $g(R_n, R_n')$ is also a function that is explicitly known.\\
Substituting (\ref{rnex}) and (\ref{beex}) into (\ref{betap}), we obtain a second-order differential equation satisfied by $R_{n}(t)$. Using the transformation (\ref{tr}), we finally get the Painlev\'{e} VI equation satisfied by $S_n(t)$.
\end{proof}

\noindent $\mathbf{Remark\: 3.}$ If $\gamma=0$, the results in Theorem \ref{thm1} and \ref{thm2} are coincident with the ones in \cite{ChenZhang}. In addition, we extend the results in \cite{Dai} to a more general situation.

\section{Double Scaling Analysis and Asymptotics}
It is well known that the standard Jacobi polynomials $P_{n}^{(\alpha,\beta)}$ admit the following important limit relationship \cite{Erdelyi} (page 173 (41))
$$
\lim_{n\rightarrow\infty}n^{-\alpha}P_{n}^{(\alpha,\beta)}\left(1-\frac{z^{2}}{2n^{2}}\right)=\left(\frac{z}{2}\right)^{-\alpha}J_{\alpha}(z),
$$
where $J_{\alpha}(z)$ is the Bessel function of the first kind of order $\alpha$. This motivates us to execute the double scaling, namely, $n\rightarrow\infty, t\rightarrow 0$, such that $s=c\: n^{2}t$ is fixed. Here $c$ is a constant to be suitably chosen later.

We define
$$
\sigma(s):=-\lim\limits_{n\rightarrow\infty}H_{n}\left(\frac{s}{c\:n^{2}}\right).
$$
After the change of variables, equation (\ref{sod}) becomes
\bea\label{ds}
&&s^2 \left(\sigma ''(s)\right)^2+4 s \left(\sigma '(s)\right)^3-\left[(\alpha+\gamma)^{2}-\frac{4s}{c}+4 \sigma (s)\right]\left(\sigma '(s)\right)^2 -\frac{2}{c}\left( \alpha  \gamma + \gamma ^2+2 \sigma (s)\right)\sigma '(s)-\frac{\gamma ^2}{c^2}\nonumber\\
&+&\frac{2(\alpha +\beta +\gamma)}{c\: n}\Big[2s \left(\sigma '(s)\right)^2- \left(\alpha\gamma+\gamma^{2}+2 \sigma (s)\right)\sigma '(s)-\frac{\gamma^{2}}{c}\Big]+O\left(\frac{1}{n^{2}}\right)=0.
\eea
We have obtained the coefficient of the $O\left(\frac{1}{n^{2}}\right)$ term. Since the expression is very long, we do not write it down here.

Let $n\rightarrow\infty$ and keep only the highest order term in (\ref{ds}), we get
$$
s^2 \left(\sigma ''(s)\right)^2+4 s \left(\sigma '(s)\right)^3-\left[(\alpha+\gamma)^{2}-\frac{4s}{c}+4 \sigma (s)\right]\left(\sigma '(s)\right)^2 -\frac{2}{c}\left( \alpha  \gamma + \gamma ^2+2 \sigma (s)\right)\sigma '(s)-\frac{\gamma ^2}{c^2}=0.
$$
By choosing $c=-1$, the above equation becomes
$$
s^2 \left(\sigma ''(s)\right)^2+4 s \left(\sigma '(s)\right)^3-\left[(\alpha+\gamma)^{2}+4s+4 \sigma (s)\right]\left(\sigma '(s)\right)^2 +2\left( \alpha  \gamma + \gamma ^2+2 \sigma (s)\right)\sigma '(s)-\gamma ^2=0,
$$
which is just the Jimbo-Miwa-Okamoto $\sigma$-form of the Painlev\'{e} III.
Hence we prove the following theorem.
\begin{theorem}
Assuming that $n\rightarrow\infty, t\rightarrow 0$, such that $s=-n^{2}t$ is fixed. Then $\sigma(s):=-\lim\limits_{n\rightarrow\infty}H_{n}\left(-\frac{s}{n^{2}}\right)$ satisfies the Jimbo-Miwa-Okamoto $\sigma$-form of the Painlev\'{e} III \cite{Jimbo1982} (see (3.13) in page 1157),
\be\label{p3}
\left(s\sigma''(s)\right)^2=4\sigma'(s)(\sigma'(s)-1)(\sigma(s)-s\sigma'(s))+(\nu_{1}\sigma'(s)-\nu_{2})^2,
\ee
with parameters $\nu_{1}=\alpha+\gamma,\; \nu_{2}=\gamma$.
\end{theorem}
Finally, we show the asymptotic behavior of the Hankel determinant in two special cases, which are related to the largest and smallest eigenvalue distribution in the degenerate Jacobi unitary ensemble.
\begin{theorem}
If $A=0,\;B=1$, then as $t\rightarrow1^{-}$,
\bea\label{asy1}
D_n(t)&\sim&2^{-2n(n+\beta+\gamma)}(2\pi)^{n}\frac{\Gamma\left(\frac{\beta+\gamma+1}{2}\right)G^2\left(\frac{\beta+\gamma+1}{2}\right)G^2\left(\frac{\beta+\gamma}{2}+1\right)}
{G(\beta+\gamma+1)G(\beta+1)G(\gamma+1)}\nonumber\\
&\times&\frac{G(n+1)G(n+\beta+1)G(n+\gamma+1)G(n+\beta+\gamma+1)}{G^2\left(n+\frac{\beta+\gamma+1}{2}\right)G^2\left(n+\frac{\beta+\gamma}{2}+1\right)
\Gamma\left(n+\frac{\beta+\gamma+1}{2}\right)}(1-t)^{n(n+\beta+\gamma)}.
\eea
If $A=1,\;B=-1$, then as $t\rightarrow 0^{+}$,
\bea\label{asy2}
D_n(t)&\sim&2^{-2n(n+\alpha+\gamma)}(2\pi)^{n}\frac{\Gamma\left(\frac{\alpha+\gamma+1}{2}\right)G^2\left(\frac{\alpha+\gamma+1}{2}\right)G^2\left(\frac{\alpha+\gamma}{2}+1\right)}
{G(\alpha+\gamma+1)G(\alpha+1)G(\gamma+1)}\nonumber\\
&\times&\frac{G(n+1)G(n+\alpha+1)G(n+\gamma+1)G(n+\alpha+\gamma+1)}{G^2\left(n+\frac{\alpha+\gamma+1}{2}\right)G^2\left(n+\frac{\alpha+\gamma}{2}+1\right)
\Gamma\left(n+\frac{\alpha+\gamma+1}{2}\right)}t^{n(n+\alpha+\gamma)}.
\eea
Here $G(x)$ is the Barnes G-function, defined by \cite{Askey}
$$
G(x+1)=(2\pi)^{\frac{x}{2}}\mathrm{e}^{-\frac{1}{2}x(x+1)-\frac{1}{2}\gamma x^{2}}\prod_{k=1}^{\infty}\left[\left(1+\frac{x}{k}\right)^{k}\mathrm{e}^{-x+\frac{x^{2}}{2k}}\right],
$$
where $\gamma$ is the Euler-Mascheroni constant.
\end{theorem}
\begin{proof}
If $A=0, B=1$, then (\ref{dnt}) becomes
$$
D_{n}(t)=\frac{1}{n!}\int_{[t,1]^{n}}\prod_{1\leq i<j\leq n}(x_{i}-x_{j})^{2}\prod_{j=1}^{n}x_{j}^{\alpha}(1-x_{j})^{\beta}(x_{j}-t)^{\gamma}dx_{j}.
$$
After the change of variables $x_{j}=y_{j}+t,\;j=1,2,\ldots,n$, the above multiple integral turns to
\bea
D_{n}(t)&=&\frac{1}{n!}\int_{[0,1-t]^{n}}\prod_{1\leq i<j\leq n}(y_{i}-y_{j})^{2}\prod_{j=1}^{n}(y_{j}+t)^{\alpha}(1-t-y_{j})^{\beta}y_{j}^{\gamma}dy_{j}\nonumber\\
&=&(1-t)^{n(n+\beta+\gamma)}\frac{1}{n!}\int_{[0,1]^{n}}\prod_{1\leq i<j\leq n}(y_{i}-y_{j})^{2}\prod_{j=1}^{n}y_{j}^{\gamma}(1-y_{j})^{\beta}[(1-t)y_{j}+t]^{\alpha}dy_{j},\nonumber
\eea
where we replace $y_j$ by $(1-t)y_j$ in the second equality.\\
As $t\rightarrow1^{-}$,
\bea
D_{n}(t)&\sim&(1-t)^{n(n+\beta+\gamma)}\frac{1}{n!}\int_{[0,1]^{n}}\prod_{1\leq i<j\leq n}(y_{i}-y_{j})^{2}\prod_{j=1}^{n}y_{j}^{\gamma}(1-y_{j})^{\beta}dy_{j}\nonumber\\
&=&2^{-n(n+\beta+\gamma)}(1-t)^{n(n+\beta+\gamma)}\frac{1}{n!}\int_{[-1,1]^{n}}\prod_{1\leq i<j\leq n}(y_{i}-y_{j})^{2}\prod_{j=1}^{n}(1-y_{j})^{\beta}(1+y_{j})^{\gamma}dy_{j}\nonumber\\
&=&2^{-n(n+\beta+\gamma)}(1-t)^{n(n+\beta+\gamma)}D_{n}[w_{\beta,\gamma}],\nonumber
\eea
where we replace the variable $y_j$ by $\frac{1+y_j}{2}$ in the second step, and $D_{n}[w_{\beta,\gamma}]$ is the Hankel determinant for the standard Jacobi weight $(1-y_{j})^{\beta}(1+y_{j})^{\gamma}$.
According to the formula (1.6) in \cite{Basor2005}, we obtain (\ref{asy1}).

On the other hand, in the case $A=1, B=-1$, from (\ref{dnt}) we have
$$
D_{n}(t)=\frac{1}{n!}\int_{[0,t]^{n}}\prod_{1\leq i<j\leq n}(x_{i}-x_{j})^{2}\prod_{j=1}^{n}x_{j}^{\alpha}(1-x_{j})^{\beta}(t-x_{j})^{\gamma}dx_{j}.
$$
By the change of variables $x_j=ty_j,\;j=1,2,\ldots,n$, it gives
$$
D_{n}(t)=t^{n(n+\alpha+\gamma)}\frac{1}{n!}\int_{[0,1]^{n}}\prod_{1\leq i<j\leq n}(y_{i}-y_{j})^{2}\prod_{j=1}^{n}y_{j}^{\alpha}(1-y_{j})^{\gamma}(1-ty_{j})^{\beta}dy_{j}.
$$
As $t\rightarrow 0^{+}$, we find
\bea
D_{n}(t)&\sim& t^{n(n+\alpha+\gamma)}\frac{1}{n!}\int_{[0,1]^{n}}\prod_{1\leq i<j\leq n}(y_{i}-y_{j})^{2}\prod_{j=1}^{n}y_{j}^{\alpha}(1-y_{j})^{\gamma}dy_{j}\nonumber\\
&=&2^{-n(n+\alpha+\gamma)}t^{n(n+\alpha+\gamma)}\frac{1}{n!}\int_{[-1,1]^{n}}\prod_{1\leq i<j\leq n}(y_{i}-y_{j})^{2}\prod_{j=1}^{n}(1-y_{j})^{\alpha}(1+y_{j})^{\gamma}dy_{j}\nonumber\\
&=&2^{-n(n+\alpha+\gamma)}t^{n(n+\alpha+\gamma)}D_{n}[w_{\alpha,\gamma}],\nonumber
\eea
where we replace the variable $y_j$ by $\frac{1-y_j}{2}$ in the second equality. Similarly, by using the formula (1.6) in \cite{Basor2005}, we obtain (\ref{asy2}).
\end{proof}

\section*{Acknowledgments}
Chao Min was supported by the Scientific Research Funds of Huaqiao University under grant number 600005-Z17Y0054.
Yang Chen was supported by the Macau Science and Technology Development Fund under grant number FDCT 023/2017/A1 and by the University of Macau under grant number MYRG 2018-00125-FST.

\end{document}